\documentclass[aap,preprint]{imsart}
\usepackage{amsmath,amssymb,url,amsfonts,mathrsfs}
\usepackage{amscd}
\usepackage{comment}
\usepackage{amsthm}
\usepackage{graphics}
\usepackage[center]{caption}

\usepackage{float,psfrag}
\usepackage{ulem}

\usepackage{color}

\long\def\jnt#1{{#1}}

\def\old#1{}

\usepackage{enumerate}
\usepackage{verbatim}
\usepackage{graphicx}

\newcommand{\be}{\begin{equation}}
\newcommand{\ee}{\end{equation}}

\newcommand{\bQ}{\mathbf{Q}}

\newcommand{\E}{\mathbb{E}}
\newcommand{\PR}{\mathbb{P}}

\newcommand{\N}{\mathbb{N}}

\newcommand{\bsigma}{{\boldsymbol{\sigma}}}
\newcommand{\sS}{{\cal{S}}}

\newcommand{\bOne}{\mathbf{1}}

\newtheorem{theorem}{Theorem}[section]
\newtheorem{lemma}[theorem]{Lemma}

\newtheorem{corollary}[theorem]{Corollary}

\providecommand{\boldsymbol}[1]{\mbox{\boldmath $#1$}}

\newcommand\T{\rule{0pt}{4ex}}       
\newcommand\B{\rule[-3ex]{0pt}{0pt}}

\begin{document}
\begin{frontmatter}

\title{On Queue-Size Scaling for Input-Queued Switches}
\maketitle

\begin{aug}
\author{\fnms{D.} \snm{Shah}\ead[label=e1]{devavrat@mit.edu}}
  \and
  \author{\fnms{J.~N.} \snm{Tsitsiklis}\ead[label=e1]{jnt@mit.edu}}
  \and
  \author{\fnms{Y.} \snm{Zhong}\thanksref{ds.ack}\ead[label=e1]{zhyu4118@mit.edu}}
    \runauthor{Shah \& Tsitsiklis \& Zhong}
  \affiliation{Massachusetts Institute of Technology}
\thankstext{ds.ack}{May 18, 2014. This work was supported by NSF \jnt{grants} CCF-0728554 \jnt{and CMMI-1234062. This research was performed while} all authors were affiliated with the Laboratory for Information and Decision Systems as well as the Operations Research Center at MIT. 
\jnt{The third author is currently with the IEOR department, at Columbia University.}
Current emails: {\tt \{devavrat,jnt\}@mit.edu}, {\tt {yz2561@columbia.edu.}}}
\end{aug}

\vspace{.1in}

\begin{abstract}

We \jnt{study} the optimal scaling of the \jnt{expected} {total} queue size in an $n\times n$ input-queued switch, as a function of the \jnt{number of ports} $n$ and the load factor $\rho$, which \jnt{has} been conjectured to be $\Theta(n/\jnt{(}1-\rho\jnt{)})$ (cf.\ \cite{STZopen}). In a recent work \cite{SWZ11}, the validity of this conjecture has been established \jnt{for the regime where} $1-\rho = O(1/n^2)$. In this paper, we make \jnt{further} progress \jnt{in the direction of} this conjecture. We provide a new class of scheduling policies under which the  {expected} total queue size scales as $O\big(n^{1.5} (1-\rho)^{-1} \log \big(1/(1-\rho)\big)\big)$ \jnt{when} $1-\rho = O(1/n)$. \jnt{This is an improvement over the state of the art;} for example, for $\rho = 1-1/n$ the best known bound was $O(n^3)$, while ours is $O(n^{2.5} \log n)$.

\end{abstract}



\end{frontmatter}

\section{Introduction}\label{sec:model}

An input-queued switch is a popular and commercially available architecture
for scheduling data packets in an  internet router. 
In general, an input-queued switch maintains a number of virtual queues 
to which packets arrive. 
\jnt{Packets to be served at each time slot are selected according to a scheduling policy, subject to} system constraints 
that specify which queues can be served simultaneously. 

The input-queued switch model is an important example of so-called 
``stochastic processing networks,'' formalized by Harrison {\cite{harrison:canonical, harrison:canonical.corr}}, 
which \jnt{have} become a canonical model \jnt{of} a variety of dynamic 
resource allocation scenarios. 
While the most basic questions concerning throughput and stability\footnote{\jnt{Under the definition that we adopt,} 
the system is {\it stable} if the \jnt{expected queue sizes are bounded over time. Furthermore,}
a policy is {\it throughput optimal} if the system is stable whenever 
\jnt{there exists some policy under which the system is stable.}
}
\jnt{are} relatively well-understood \jnt{for} general  stochastic processing networks {(see e.g., \cite{LD05}, \cite{K-M}, \cite{KW04}, \cite{daibala}, \cite{TE92}, \cite{MAW96}, \cite{DLK01})}, 
\jnt{much less is known on the subject of {more refined performance measures} (e.g., results on the distribution and the} 
moments of queue sizes),  
even \jnt{for} the special context of input-queued switches.

This paper contributes \jnt{to the performance analysis of} stochastic processing networks. \jnt{It} is motivated by the conjectures put forth in \cite{STZopen} 
\jnt{on} the optimal scaling of the \jnt{expected} total queue size 
in an $n\times n$ input-queued switch, 
as a function of the \jnt{number of ports} $n$ and the load factor $\rho$. 
\jnt{For certain} limiting regimes, it was conjectured \jnt{in} \cite{STZopen} that 
the optimal scaling \jnt{(that is, the scaling under an ``optimal'' policy)} takes the form $\Theta\left(n/(1-\rho)\right)$. \jnt{This is to be compared to available results that}
include an  $\jnt{O(}n^2/(1-\rho))$ upper bound, 
achieved by the so-called Maximum-Weight policy \cite{shahkopi}, \cite{leonardi-delay}, 
and an  $O(n\log n/(1-\rho)^2)$ upper bound, achieved by 
a {batching} policy proposed in \cite{NMC07}. 
\jnt{More} recently, Shah et al.\ \cite{SWZ11} proposed a policy that gives an 
upper bound of
$\frac{n}{1-\rho} + n^3$, thus \jnt{establishing the validity of the} conjecture when $1-\rho = O(1/n^2)$.

In this paper, we 
\jnt{focus on a different regime, where} 
\jnt{$1/n^2 \ll 1-\rho \leq 1/n$.} 
\jnt{In some sense, this is a more difficult regime to analyze, when compared to the regime where $1-\rho = O(1/n^2)$. This is because we consider a larger ``gap'' $1-\rho$, and so the heavy-traffic aspects of the system are less pronounced. This in turn means that various laws of large numbers (e.g., fluid or batching arguments) are less effective.} 

Concretely, we shall focus on \jnt{the case} $\rho = 1 - 1/\jnt{f_n}$, \jnt{where $f_n \geq n$ for all $n$, and for $n$ tending to infinity.} When $f_n = n$, previous works 
give an upper bound $O(n^3)$ (ignoring poly-logarithmic dependence on $n$) 
on the \jnt{expected} total queue size. 
In contrast, when $\rho = 1 - 1/n$, the conjectured optimal scaling $O\left(n/(1-\rho)\right)$ \jnt{is of the form}
$O(n^2)$. 
It is then natural to ask whether this gap can be reduced, i.e., 
whether there exists a policy under which 
the \jnt{expected} total queue size is upper bounded by 
$O(n^{\alpha})$, with $\alpha < 3$ (and ideally \jnt{with} $\alpha=2$), 
when $\rho = 1- 1/n$. 

Our main contribution is a new policy
\jnt{that leads to an upper bound of}
$O\big(n^{1.5} f_n \log f_n\big)$, 
\jnt{when $f_n\geq n$ and the arrival rates at the different queues are all equal.}
As a corollary, \jnt{if} $f_n = n$, 
the \jnt{expected} total queue size 
is upper bounded by $O(n^{2.5}\log n)$. 
This is the best known scaling with respect to $n$, when $\rho = 1 - 1/n$. 
While this is a significant improvement over existing bounds, 
we \jnt{still} believe that the right scaling \jnt{(ignoring any poly-logarithmic factors)} is $O(n^2)$. 
The \jnt{best currently} known scalings on the {expected} total queue size 
under various regimes, 
in an $n\times n$ input-queued switch, are summarized in Table \ref{tab:bounds}.

\begin{table}[h]
\begin{center}
\caption{Best known scalings \jnt{of the expected total queue size} in various regimes. \jnt{Here, $\rho$ is the} load factor \jnt{and $n$ is the} number of \jnt{input} ports.}
\vspace{.5cm}
\begin{tabular}{|c|c|c|}
\hline
Regime & Scaling & References \rule{0pt}{2.6ex} \rule[-1.2ex]{0pt}{0pt}
\\
\hline
$\frac{1}{1-\rho} < n$ & $O\left(\frac{n \log n}{(1-\rho)^2}\right)$ & \cite{NMC07} \T\B
 \\
\hline
$\frac{1}{1-\rho} = n$ & $O\left(n^{2.5}\log n\right)$ & this work \T\B\\
\hline
$n\leq \frac{1}{1-\rho} < n^2$ & $O\left(\frac{n^{1.5}\log n}{1-\rho}\right)$ & this work \T\B\\
\hline
$\frac{1}{1-\rho} \geq n^2$ & $\Theta\left(\frac{n}{1-\rho}\right)$ & \cite{SWZ11} \T\B\\
\hline
\end{tabular}
\label{tab:bounds}
\end{center}
\end{table}

The policy that we propose is a variation of the standard batching policy. 
In the standard batching policy, time
is divided into disjoint intervals or batches. 
Packets that arrive in a given batch are served only after the arrival of the entire batch. 
By choosing the batch length large enough (deterministically or randomly), the  total \jnt{number of arriving packets is close to its expected value and} can be served efficiently. 
\jnt{In general, a longer batching interval improves efficiency, 
because the effect of random fluctuations is less pronounced, but on the other hand leads to larger delays and queue sizes.}
 \jnt{For this reason,} a good batching policy, \jnt{as} for 
example in \cite{NMC07}, selects the \jnt{smallest} possible batch  {length} that 
\jnt{will guarantee stability;} 
in 
\cite{NMC07}, this led to \jnt{a} bound of $O\left(\frac{n \log n}{(1-\rho)^2}\right)$ on \jnt{the expected total}  queue size. 

\jnt{Given the stability requirement, we cannot hope to improve delay by reducing the batch
 {length}. On the other hand, the policy that we consider}
starts serving \jnt{packets} from a given batch 
a lot earlier, 
\jnt{before the arrival of the entire batch.} By starting to serve early, the \jnt{expected}  
delay (and hence queue size) is reduced. \jnt{When the arrival rates at each queue are all equal,} 
we show that the arrival process has sufficient regularity at \jnt{a} time scale shorter than 
the batch  {length}. 
\jnt{Consequently,} the policy can {indeed} start serving the arriving \jnt{packets} early,
\jnt{while making sure that the stochastic fluctuations lead to only a small number of unserved packets, which can be ``cleared'' efficiently at the end of the batch.}
\jnt{The combination of these ideas} results in substantial 
improvement over the standard batching policy. 


\jnt{A} few remarks are in order regarding the proposed policy. {Our} policy \jnt{relies on {the} assumption of {uniform} arrival rates}.  In contrast, \jnt{some} existing policies, such as the maximum weight policy
or \jnt{the one} in \cite{SWZ11}, are based {only} on the observed system state (the queue \jnt{sizes}) \jnt{and are effective even with non-uniform arrival rates.} 
\jnt{However,} we believe that our policy and its analysis 
can be modified to account for general \jnt{(non-uniform)} arrival rates.

\subsection{Organization}

The rest of the paper is organized as follows.  In Section \ref{sec:model}, we describe the 
input-queued switch model.
In Section \ref{sec:main}, we state our main theorem.  In Section \ref{sec:prelim}, 
we introduce some preliminary 
facts and theorems, which \jnt{will} be used in later sections. In Section \ref{sec:policy}, we describe our policy. In 
Section \ref{sec:analysis}, we provide the proof of the main theorem. We conclude with some discussion in Section \ref{sec:discussion}.

\section{Input-queued switch model}\label{sec:model}

An $n\times n$ 
input-queued switch has $n$ input ports and $n$ output ports. 
The switch operates in discrete time, indexed by $\tau \in \jnt{\{1, 2, \dots\}}$.
In each time slot, and for \jnt{each} port pair $(i,j)$, \jnt{a} unit-sized \jnt{packet} may arrive at input port $i$ destined
for output port $j$, according to an exogenous arrival process. Let $A_{i,j}(\tau)$
denote the cumulative number of such arriving packets \jnt{during time slots  $1,\ldots,\tau$.} 
We assume that the processes $A_{i,j}(\cdot)$ are independent
for different pairs $(i,j)$. Furthermore, for every input-output pair $(i,j)$, 
{$\{A_{i,j}(\tau) - A_{i,j}(\tau-1)\}_{\tau\in \N}$} is a Bernoulli process 
with parameter \jnt{$\rho/n$}{, with the convention that $A_{i,j}(0)=0$.} In particular, 
$$
\jnt{\E[}A_{i,j}(\tau)] = {\jnt{\frac{\rho}{n}\tau}},\quad 
\jnt{\mbox{for \ all\ }i,j,\ \mbox{and all }\tau\geq 1.}
$$
\jnt{We are only interested in systems that can be made stable under a suitable policy,  and for this reason, we assume that $\rho<1$, i.e., that the system is underloaded.}
\jnt{Furthermore,} we consider \jnt{a} system load $\rho$ \jnt{of the form} $\rho = 1 - 1/f_n$, 
where \jnt{the sequence \{$f_n$\}  satisfies $f_n \geq n$ for all $n$.}  

For every input-output pair $(i,j)$, the associated arriving packets are stored in separate
queues, so that we have a total of $n^2$ queues. Let $Q_{i,j}(\tau)$ be 
the number of packets waiting at input port $i$, destined for output {port} $j$, at the 
beginning of time slot $\tau$. 

\old{We remark here that, in the sequel, we will often consider the total number of packets that are waiting 
at an input port, or that are destined for the same output port. 
More precisely, suppose that we have a queue matrix $\bQ = \left(Q_{i,j}\right)$. 
For each $i$, let $R_i = \sum_{j'=1}^n Q_{i,j'}$ be the $i$th row sum, 
and for each $j$, let $C_j = \sum_{i'=1}^n Q_{i',j}$ be the $j$th column sum. 
Then $R_i$ is also the total number of packets that are waiting at input port $i$, 
and, similarly, $C_j$ the total number of packets that are destined for output port $j$. 
For convenience, we sometimes use ``row sum'' to mean the total number 
of packets that are waiting at an input port, and ``column sum'' to mean 
the total number of packets that are destined for an output port. }

In each time slot, the switch can transmit a number of packets 
from input ports to output ports, subject to the following two constraints: 
(i)  each input port can transmit at most one packet; and,  
(ii) each output port can receive at most one packet. In other words, 
the actions of a switch at a particular time slot constitute a
{\it matching} between input 
and output ports.

A matching, or {\it schedule}, can be described by \jnt{an array} 
$\bsigma\in\{0,1\}^{n\times n}$, where $\sigma_{i,j}=1$ if 
input port $i$ is matched to output port $j$, and $\sigma_{i,j}=0$ otherwise. 
Thus, \jnt{at any given time,} the set of all feasible schedules is
\[
\sS = \Big\{\bsigma \in \{0,1\}^{n \times n} : \sum_k \sigma_{i,k} \leq 1, ~\sum_k \sigma_{k, j} \leq 1, ~\forall \ (i,j) \mbox{ with }1\leq i, j \leq n\Big\}.
\]
A scheduling policy (or simply {\it policy}) is a rule that, at any given time $\tau$, chooses a schedule $\bsigma(\tau) = [\sigma_{i,j}(\tau)] \in \sS$, based on \jnt{the} past history and the current \jnt{queue sizes $Q_{i,j}(\tau)$.} 
If $\sigma_{i,j}(\tau) = 1$
and $Q_{i,j}(\tau) > 0$, then one packet is removed from the queue associated with the pair $(i,j)$.

Regarding the details of the model, we adopt the following \jnt{timing} conventions.
At the beginning of time slot $\tau$, the queue \jnt{sizes $Q_{i,j}(\tau)$ are} observed  by the policy. The schedule $\bsigma(\tau)$ is applied
in the middle of the time slot. Finally, at the end of the time slot, new arrivals happen. Mathematically, for all
$i$, $j$, and $\tau \in \N$, we have
\begin{equation}
Q_{i,j}(\tau+1) = Q_{i,j}(\tau) - \sigma_{i,j}(\tau) \bOne_{{\{}Q_{i,j}(\tau) > 0{\}}} + A_{i,j}(\jnt{\tau}) - A_{i,j}(\tau\jnt{-1}),
\label{eq:dyn}
\end{equation}
where for a set $B$, $\bOne_{B}$ is its indicator function. 
We assume throughout the paper that the system starts empty, 
i.e., $Q_{i,j}(\jnt{1}) = 0$, for all $i, j$. 

Summing Eq. (\ref{eq:dyn}) over time {and \jnt{using the assumption} $Q_{i,j}(1) = 0$, we get the following equivalent expression, for $\tau \in \N$:
\begin{equation}
Q_{i, j}(\tau\jnt{+1}) =A_{i, j}(\tau) - \sum_{t = \jnt{1}}^{\jnt{\tau}} \sigma_{i,j}(t)\bOne_{{\{Q_{i,j}(t) > 0\}}}.
\label{eq:dyncum}
\end{equation}
\jnt{We define} 
$$S_{i,j}(\tau)=
\sum_{t = \jnt{1}}^{\jnt{\tau}} \sigma_{i,j}(t)\bOne_{\{Q_{i,j}(t) > 0\}},$$
so that (\ref{eq:dyncum}) reduces to
\[
Q_{i,j}(\tau\jnt{+1}) = A_{i,j}(\tau) - S_{i,j}(\tau).
\]
We call $S_{i,j}(\tau)$ the \jnt{\it actual}\/ service offered to queue $(i,j)$
\jnt{during the first $\tau$ time slots. Note that}
$S_{i, j}(\tau)$ \jnt{may be} different from $\sum_{t = 1}^{\jnt{\tau}} \sigma_{i,j}(t)$, \jnt{which is} 
the cumulative service {\it offered} to queue $(i,j)$ \jnt{during the first $\tau$ slots.}

\section{Main Result}\label{sec:main}

\jnt{The main result of this paper is as follows.}

\begin{theorem}\label{thm:main}
Consider an $n\times n$ input-queued switch 
\jnt{in which} the arrival processes are independent Bernoulli processes 
\jnt{with a common} arrival rate $\rho/n$, \jnt{where} $\rho = 1 - 1/f_n$ \jnt{and} $f_n \geq n$. 
\jnt{For any {$n$}, there exists} a scheduling policy under which the \jnt{expected} total queue size 
is upper bounded by $cn^{1.5} f_n \log f_n$. That is, 
\[
\sum_{i,j=1}^n \jnt{\E[}Q_{i,j}(\jnt{\tau})]
\leq cn^{1.5} f_n\log f_n, \qquad \jnt{\mbox{for all }\tau,}
\]
where $c$ is a constant that does not depend on $n$.
\end{theorem}

\begin{corollary}\label{cor:main}
Consider the setup in Theorem \ref{thm:main}, \jnt{with} $f_n = n$. 
\jnt{For any {$n$}, there exists} a scheduling policy under which the \jnt{expected} total queue size 
is upper bounded by $cn^{2.5} \log n$. That is, 
\[
\sum_{i,j=1}^n \jnt{\E[}Q_{i,j}(\jnt{\tau})] 
\leq cn^{2.5} \log n, \qquad \jnt{\mbox{for all }\tau,}
\]
where $c$ is a constant that does not depend on $n$.
\end{corollary}

{Let us remark here that  we {only} prove Theorem \ref{thm:main} for all sufficiently large $n$. 
The validity of the theorem for smaller $n$ is guaranteed}
{by considering an arbitrary stabilizing policy (e.g., the maximum weight policy) and letting $c$ be large enough so that we have an upper bound to the expected total queue size under that policy.}

\section{Preliminaries}\label{sec:prelim}
Here we state some facts \jnt{that} will be \jnt{used in our} subsequent analysis. 
\paragraph{Concentration Inequalit\jnt{ies}}
We \jnt{will use} the following \jnt{tail bounds for binomial random variables}  ({adapted from} Theorem 2.4 in \cite{FCgraph}).
\begin{theorem}\label{thm:concent}
Let $X_1, X_2, \ldots, X_m$ be independent and identically distributed Bernoulli random variables, with 
\[
\PR(X_i = 1) = p, \quad \mbox{ and } \quad \PR(X_i = 0) = 1-p,
\]
for $i = 1, 2, \ldots, m$. 
Let $X = \sum_{i=1}^m X_i$, \jnt{so that} $\E[X] = mp$. Then, for any $x > 0$, we have
\begin{eqnarray}
(Lower\  tail) \quad \quad \PR(X \leq \E[X] - x) & \leq & \exp\left\{-\frac{x^2}{2\E[X]}\right\}, \label{eq:lower}\\
(Upper\  tail) \quad \quad \PR(X \geq \E[X] + x) & \leq & \exp\left\{-\frac{x^2}{2(\E[X]+x/3)}\right\}. \label{eq:upper}
\end{eqnarray}
\end{theorem}

\paragraph{Kingman Bound for \jnt{the discrete-time} $G/G/1$ Queue} 
Consider a discrete-time $G/G/1$ queueing system. 
More precisely, let $\jnt{X}(\tau)$ be the number of packets that arrive during time slot $\tau$, 
 let $\jnt{Y}(\tau)$ be the number of packets that can be served \jnt{during slot} $\tau${, 
and let $Z(\tau)$ be the queue size at the beginning of time slot $\tau$}. Suppose 
that \jnt{the} $X(\tau)$ are i.i.d.\ across time, \jnt{and} so are \jnt{the} $Y(\tau)$. 
Furthermore, \jnt{the processes $X(\cdot)$ and $Y(\cdot)$ are independent.} The queueing dynamics \jnt{are} given by
\begin{equation}\label{eq:qdyn}
\jnt{Z}(\jnt{\tau}+1) = \max\{0, Z(\tau) + X(\tau) - Y(\tau)\}.
\end{equation}
Let $\lambda = \E[X(\jnt{\tau})]$, $m_{2x} = \E[X^2(\tau)]$, $\mu = \E[Y(\tau)]$, and 
$m_{2y} = \E[Y^2(\tau)]$. Suppose that $\lambda < \mu$. The following bound is proved in 
 \cite{SYbook2014} {(Theorem 3.4.2)}, \jnt{using a standard argument based on a  quadratic Lyapunov function.}
\begin{theorem}[Discrete-time Kingman bound]\label{thm:kingman}
\jnt{Suppose that $Z(1)=0$ and that $\lambda<\mu$. Then, }
\begin{equation}\label{eq:kingman-discrete}
\E[Z(\tau)] \leq \frac{m_{2x} + m_{2y} - 2\lambda\mu}{2(\mu-\lambda)},
\qquad\jnt{\mbox{for all }\tau.}
\end{equation}
\end{theorem}
\jnt{In fact, the above theorem is proved in \cite{SYbook2014} for the expected queue size in steady state. However, since we assume that $Z(1)=0$, a standard coupling argument shows that the same bound holds for $\E[Z(\tau)]$ at any time $\tau$.}

\paragraph{Optimal Clearing Policy} 
\jnt{Similar to \cite{NMC07}, we will use} the concept of the {\it minimum clearance time} of a queue matrix.
Consider a certain queue matrix $[Q_{i,j}]_{i,j=1}^n$, where $Q_{i,j}$ denotes the number 
of packets at input port $i$ destined for output port $j$. 
Suppose that no new packets \jnt{arrive, and that} the goal is to simply clear all packets present in the system, 
in {the least possible} {amount of} time, using only feasible schedules/matchings. We call this \jnt{minimal required time} the 
{\it minimum clearance time} of the \jnt{given} queue matrix, and we denote it by \jnt{$L$.} 
Then, \jnt{$L$} is characterized exactly as follows. 

\begin{theorem}
\label{thm:bvn}
Let $[Q_{i,j}]_{i,j=1}^n$ be a queue matrix. 
Let 
$$R_i = \sum_{j = 1}^n Q_{i,j}\qquad \jnt{\mbox{and}}\qquad  C_j = \sum_{i=1}^n Q_{i, j}$$ be the \jnt{$it$h row sum and the} $j$th column sum, \jnt{respectively.}
\jnt{Then, the minimum clearance time, $L$, is equal to the largest of the row and column sums:}
\begin{equation}\label{eq:min-clear}
{L = \max \left\{\max_i R_i, \max_j C_j\right\}.}
\end{equation}
\end{theorem}

{The proof of Theorem \ref{thm:bvn} is a simple modification of {the proof} of Theorem 5.1.9 in \cite{Hall1998}.}

\jnt{Note that} in each time slot at most one packet can depart from each input/output port, \jnt{and therefore} 
each $R_i$ \jnt{and} $C_j$ is decreased by at most $1$. Thus, 
the minimum clearance time \jnt{cannot be smaller than the right-hand side of} \eqref{eq:min-clear}. 
Theorem \ref{thm:bvn} states that \jnt{there actually exists an {\it optimal clearing policy}} that clears all packets \jnt{within exactly {$\max \left\{\max_i R_i, \max_j C_j\right\}$}}
time slots. 

\section{Policy Description}\label{sec:policy}
To describe our policy, we \jnt{introduce} three parameters, $b$, $d$, and $s$, 
which specify \jnt{the lengths of certain} time intervals, \jnt{and which, in turn, delineate the different phases of the policy.} 
They are given by\footnote{\jnt{We will treat these parameters as if they were guaranteed to be integer{s}. Rounding them up or down to a nearest integer would overburden our notation but would have no effect on our order-of-magnitude estimates.}}
\begin{eqnarray}
b &=& c_b f_n^2 \log f_n, \label{eq:T} \\
d &=& c_d \sqrt{n} f_n \log f_n, \label{eq:S} \\
s &=& \rho b + \sqrt{c_s b \log f_n}. \label{eq:Y}
\end{eqnarray}
\jnt{Without loss of generality, we will always assume that $n\geq 3$, so that $\log f_n>{1}$. Here $c_b$, $c_d$, and $c_s$ are positive constants (independent of $n$) that will be appropriately chosen. As will be seen in the course of the proof, it suffices to choose them so that 
\begin{equation}\label{eq:c}
c_b>c_s, \qquad c_d^2\geq 640 c_b, \qquad c_d> c_b,\qquad c_s\geq 30,
\end{equation}
and which we henceforth assume. 
We note that the above inequalities do {not} necessarily lead to the best choices for these constants but they are imposed in order to simplify the details of the proof.
}

{For an $n\times n$ input-queued switch, we} 
also define $n$ \jnt{particular} schedules $\bsigma^{(1)}, \bsigma^{(2)}, \ldots, \bsigma^{(n)}$. 
For $m \in \{1,2, \ldots, n\}$, $\bsigma^{(m)}$ is defined by
\[
\sigma^{(m)}_{i,j} = \left\{ \begin{array}{ll}
1, & \mbox{ if } \jnt{j=i+m-1\ \ (\mbox{modulo }n),} \\
0, & \mbox{ otherwise. }
\end{array} \right.
\]
To illustrate, when $n = 3$, the schedules $\bsigma^{(1)}, \bsigma^{(2)}$, and $\bsigma^{(3)}$ are given by
\[
\bsigma^{(1)} = \left( \begin{array}{ccc}
1 & 0 & 0 \\
0 & 1 & 0 \\
0 & 0 & 1
\end{array}
\right), \quad
\bsigma^{(2)} = \left( \begin{array}{ccc}
0 & 1 & 0 \\
0 & 0 & 1 \\
1 & 0 & 0
\end{array}
\right), \quad
\mbox{ and } \quad \bsigma^{(3)} = \left( \begin{array}{ccc}
0 & 0 & 1 \\
1 & 0 & 0 \\
0 & 1 & 0
\end{array}
\right).
\]
Note that 
\[
\bsigma^{(1)} + \bsigma^{(2)} + \cdots + \bsigma^{(n)}  = 
\left(\begin{array}{ccc}
1 & \cdots & 1 \\
\vdots & \ddots & \vdots \\
1 & \cdots & 1
\end{array}
\right),
\]
the $n \times n$ matrix of all $1${s}.


We now proceed with the description of the policy. Time is divided into consecutive intervals, which we call {\it arrival periods}, of length $b$. For $k=0,1,2,\ldots$, the $k$th arrival period consists of slots $kb+1,kb+2,\ldots,(k+1)b$. Arrivals that occur during the $k$th arrival period are said to belong to the $k$th {\it batch.}

The general idea behind the policy is as follows. The policy aims to serve all of the packets in the $k$th batch during the $k$th {\it service period}, of length $b$, which is offset from the arrival period by a delay of $d$. Thus, the $k$th service period consists of time slots $kb+d+1,\ldots,(k+1)b+d$. If the policy does not succeed in serving all of the packets in the $k$th batch, the unserved packets will be considered {\it backlogged} and will be handled together with newly arriving packets from subsequent batches, in subsequent service periods. As it will turn out, however, the number of backlogged packets will be zero, with high probability.

We now continue with a precise description, by considering what happens during the $k$th service period. Note that the time slots $kb+1,\ldots,kb+d$ do not belong to the $k$th service period. Packets from the $k$th batch will accumulate during these time slots, but none of them will be served. At the beginning of the $k$th service period (the beginning of time slot $bk+d+1$), we may have some backlogged packets from previous service periods, and we denote their number by $B_k$. We assume that $B_0=0$.

The $k$th service period consists of three phases, which are described below and are illustrated in Fig.~\ref{f:f}.
\begin{figure}[ht]
\includegraphics[width=4.5in]{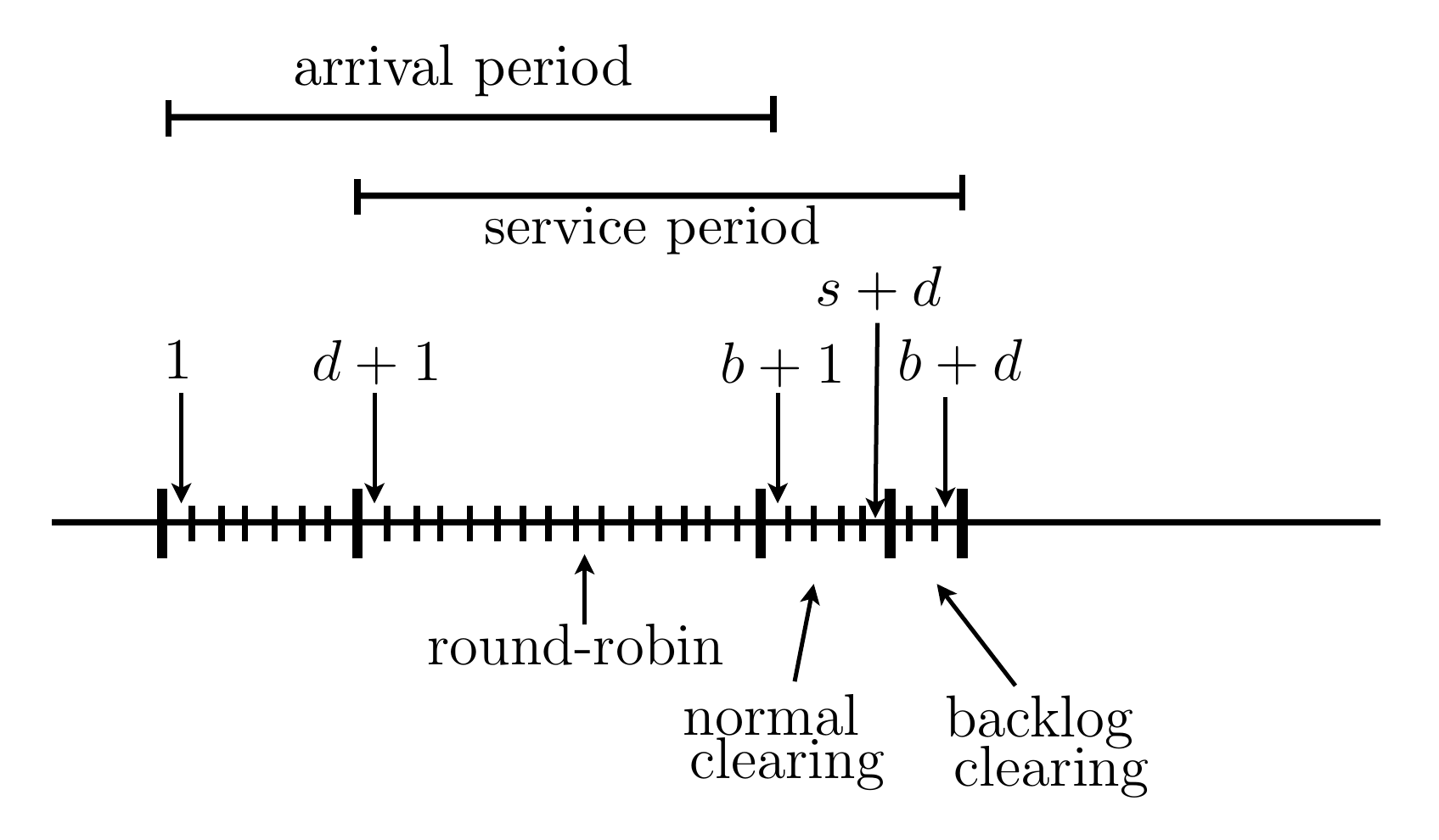}
\caption{Illustration of a typical arrival {period} 
and the phases of a service {period}. 
Slots are numbered consecutively, starting with the first slot of the arrival {period}.}
\label{f:f}
\end{figure}

\begin{enumerate}
\item
The first $b-d$ slots of the $k$th service period, namely, slots $kb+d+1,\ldots,(k+1)b$,  comprise a {\it round-robin phase}: we 
cycle through 
the schedules 
$\bsigma^{(1)}$, $\bsigma^{(2)}$, \ldots, $\bsigma^{(n)}$ in a round-robin manner.
However, during this phase, we do not serve any of the backlogged packets; we only serve packets that belong to the $k$th batch.\footnote{This particular choice introduces some inefficiency, because offered service will be wasted  whenever a queue has backlogged packets but no packets that belong to the $k$th batch. However, this choice simplifies our analysis and makes little actual difference, because the number of backlogged packets is zero with high probability.}

\item
The next $\ell=d+s-b$ slots, namely slots $(k+1)b+1,\ldots,kb+d+s$, comprise the $k$th {\it normal clearing phase}. 
Similar to the round-robin phase, we do not serve any backlogged packets during this phase. Furthermore, even though packets from the $(k+1)$st batch may have started to arrive, we do not serve any of them.
By the beginning of this phase, all of the arrivals from the $k$th batch have already arrived. Some of them have already been served during the round-robin phase. To those that remain, we apply the optimal {clearing} policy described earlier{; cf. Theorem \ref{thm:bvn}}. However, there is a possibility that the phase terminates before we succeed in serving all of the remaining packets from the $k$th batch. Let $U_k$ be the number of the packets from the $k$th batch that were left unserved during this phase. These $U_k$ packets are considered backlogged and are added to the backlog $B_k$ from earlier periods.

\item
The last $r=b-s$ slots, namely slots $kb+d+s+1,\ldots, {(k+1)b+d}$, comprise the $k$th {\it backlog clearing phase}. During this phase, we serve backlogged packets using some arbitrary policy. The only requirement is that the policy {serve} at least one packet  {at each slot that} a backlogged packet is available. However, we do not serve any of the newly arrived packets from the $(k+1)$st batch.
Any backlogged packets that are not served during this phase remain backlogged and comprise the number $B_{k+1}$ of backlogged packets at the beginning of the next service period. Since at least one backlogged packet is served (whenever available) during each one of these $r$ slots, and since there are no additions to the backlog during this phase, we have
\begin{equation}\label{eq:bevol}
B_{k+1}\leq \max\{0,B_k+U_k-r\},\qquad k=0,1,\ldots
\end{equation}
\end{enumerate}

The total length of the three phases is
$$(b-d)+ (d+s-b)+(b-s)=b,$$
so that the length of a service period is equal to the length of an arrival period. However,   
before continuing, we need to make sure that the duration of each phase is a positive number, so that the policy is well-defined. This is accomplished in the next two lemmas, which also provide order of magnitude information on the durations of these phases.

\begin{lemma}\label{l:first}
The length $r=b-s$ of the backlog clearing phase satisfies
$$r=c_r f_n \log f_n,$$
where
$c_r=c_b-\sqrt{c_s c_b}>0$.  In particular, when $n$ is large enough, we have $r\geq 1$.
\end{lemma}

\begin{proof}
Using the assumption $\rho=1-1/f_n$, we have $(1-\rho)b=b/f_n=c_b f_n\log f_n$. We then obtain
\begin{eqnarray*}
b-s&=& b-\rho b -\sqrt{c_s b \log f_n}\\
&=&c_b f_n \log f_n -\sqrt{c_s c_b f_n^2\log^2 f_n}\\
&=&(c_b-\sqrt{c_s c_b})f_n \log f_n\\
&=&c_r f_n \log f_n.
\end{eqnarray*}
The fact that $c_r>0$ follows from our assumption in Eq.\ \eqref{eq:c}.
\end{proof}

\begin{lemma} 
The length $\ell=d+s-b$ of the normal clearing phase satisfies
$$\ell \geq c_{\ell}\sqrt{n} f_n \log f_n,$$
where
$c_{\ell}=c_d-c_r>0$. In particular, when $n$ is large enough, we have $\ell\geq 1$.

\end{lemma}

\begin{proof}
Recall that $r=b-s$. It follows that
\begin{eqnarray*}
\ell&=&d+s-b\\
&=&d-r\\
&=&c_d \sqrt{n} f_n \log f_n -c_r f_n \log f_n\\
&\geq& (c_d-c_r) \sqrt{n} f_n\log f_n\\
&=&c_{\ell} \sqrt{n} f_n\log f_n.
\end{eqnarray*}
Note that $c_r<c_b<c_d$ (cf.\ Lemma \ref{l:first}  and Eq.\ \eqref{eq:c}), which implies that $c_{\ell}>0$.
\end{proof}

\section{Policy Analysis}\label{sec:analysis}
The performance analysis of the  proposed policy involves the following line of argument for what happens during the $k$th arrival and service period.
\begin{enumerate}
\item[(a)] 
In the first $d$ slots of the $k$th arrival period, we have an expected number $O(nd)$ of arrivals.
\item[(b)] With high probability, {at every time slot} during the round-robin phase, there is a positive number of packets from the $k$th arrival batch at each queue; cf.\ Lemmas \ref{l:aa} and \ref{l:w}. Therefore, offered service is never wasted. In particular, at least as many packets are served as they arrive (in the expected value sense), and the total queue size does not grow.
\item[(c)] With high probability, all of the packets from the $k$th batch that are in queue at the beginning of the normal clearing phase get cleared and therefore the number $U_k$ of newly backlogged packets is zero; cf.\ Lemma \ref{wh}.
\item[(d)] The number $B_k$ of backlogged packets evolves similar to a discrete-time G/D/1 queue; cf.\ Eq.~\eqref{eq:bevol}. Because $U_k$ is zero with high probability, the Kingman bound (Theorem \ref{thm:kingman}) implies that the expected number of backlogged packets, at any time, is small; cf.\ Lemma~\ref{l:back}.
\end{enumerate}
The above steps, when translated into precise bounds on queue sizes, will lead to an $O(nd)$ bound on the expected total queue size at any time.

\subsection{No waste during the round-robin phase}
In this subsection, we establish that during the round-robin phase, every queue contains a nonzero number of packets from the current arrival batch, with high probability. We first introduce some convenient notation. We will use the variable $t\in\{1,\ldots,b+1\}$ to index the $b$ slots of the $k$th arrival period together with the first slot of the subsequent normal clearing phase. For $t\in\{1,\ldots,b\}$, we
let $A_{i,j}^k(t)$ be the number of arrivals to the $(i,j)$th queue during the first $t$ time slots of the $k$th arrival period; these are the time slots $kb+1,kb+{2},\ldots,kb+t$. 
Similarly, for $t\in\{1,\ldots,b\}$, we let $S_{i,j}^k(t)$ be the number of packets that arrive to queue $(i,j)$ during the $k$th arrival period and get served during the first $t$ time slots of the $k$th arrival period. Finally, for $t\in\{1,\ldots,b+1\}$, we
let $Q_{i,j}^k(t)$ be the number of packets from the $k$th arrival batch that are in queue $(i,j)$ at the beginning of the $t$th slot of the $k$th arrival period. 
With these definitions, we have, 
\begin{equation}\label{eq:qk}
Q^k_{i,j}(t+1)=A^k_{i,j}(t)-S^k_{i,j}(t), \qquad t=1,\ldots,b.
\end{equation}

We are interested in conditions under which no offered service is wasted during the round-robin phase. 
Equivalently, we are interested in conditions under which all queues have a positive number of packets from the $k$th batch.
Note that the round-robin phase involves slots for which $t\in\{d+1,\ldots,b\}$. We have the following observation on the queue sizes at the beginning of these slots.

\begin{lemma}\label{l:aa}
Suppose that {$t\in\{d,\ldots,b-1\}$} 
and that 
$$A_{i,j}^k(t)> \frac{t-d}{n} +1.$$
Then, $Q^k_{i,j}(t+1)>0$.
\end{lemma}
\begin{proof} 
Note that that for the first $d$ time slots, packets from the $k$th batch do not receive any service.
Starting from the $(d+1)$st slot, we are in the round-robin phase, and queue $(i,j)$ is offered service once every $n$ slots. Therefore,
$$S^k_{i,j}(t)\leq \Big\lceil \frac{t-d}{n} \Big\rceil < \frac{t-d}{n} +1<A_{i,j}^k(t).$$
The result follows from Eq.\ \eqref{eq:qk}.
\end{proof}

The previous lemma highlights the importance of the events $A_{i,j}^k(t)> {(t-d)/n+1}$. We will show that the complements of these events have, collectively, small probability. To this effect, let $W_{i,j}^k(t)$ be the event defined by
$$W^k_{i,j}(t) =\Big\{ A_{i,j}^k(t)\leq \frac{t-d}{n} +1\Big\},\qquad
{t=d,\ldots,b-1.}$$
Let also {$W^k$} be the union of these events, over all queues, and over all indices $t$ that are relevant to the round-robin phase:
$$W^k=\bigcup_{i=1}^n \bigcup_{j=1}^n \bigcup_{{t=d}}^{{b-1}} W^k_{i,j}(t).$$

\begin{lemma}\label{l:w}
For $n$ sufficiently large,  we have
$$\PR(W^k)\leq \frac{1}{2f_n^{13}},\qquad\mbox{for all }k.$$
\end{lemma}
\begin{proof}
Let us fix some $(i,j)$ and some $t\in\{d,\ldots,b{-1}\}$. 
Note that $\E\big[A_{i,j}^k(t)\big]=\rho t/n$. Therefore, the event $W_{i,j}^k(t)$ is the same as the event
$$\Big\{A_{i,j}^k(t)\leq \E\big[A_{i,j}^k(t)\big] -\frac{\rho t}{n}+\frac{t-d}{n} +1\Big\},$$
which is of the form
$$\Big\{A_{i,j}^k(t)\leq \E\big[A_{i,j}^k(t)\big] -x\Big\},$$
where
\begin{eqnarray*}
x&=&\frac{\rho t}{n}-\frac{t-d}{n} -1\\
&=&\frac{\rho (t-d)}{n}-\frac{t-d}{n}+\frac{\rho d}{n} -1\\
&=&-(1-\rho)\frac{t-d}{n} +\frac{\rho d}{n} -1.
\end{eqnarray*}
Using the facts $t-d\leq b$ and $1-\rho=1/f_n$, the first term on the right-hand side is bounded above (in absolute value) by $b/(nf_n)$. For the second term, we use the facts  $\rho=1-(1/f_n)$, $f_n\geq n\geq 2$, to obtain $\rho\geq 1/2$.  Therefore, 
\begin{eqnarray*}
x&\geq& -\frac{b}{nf_n}+\frac{d}{2n}-1\\
&=&\frac{1}{n}\Big((c_d/2)\sqrt{n}f_n\log f_n  -c_b f_n \log f_n -n\Big)\\
&\geq&\frac{1}{n}\Big((c_d/2)\sqrt{n}f_n\log f_n  -(c_b+1) f_n \log f_n \Big).
\end{eqnarray*}
Now, for $n$ large enough, we have
$c_b+1\leq (c_d/4)\sqrt{n}$, and this implies that
\begin{equation}\label{eq:x}
x\geq \frac{1}{n}\cdot \frac{c_d}{4}\cdot \sqrt{n} f_n\log f_n =\frac{c_d f_n \log f_n}{4 \sqrt{n}}.
\end{equation}

Using Eq.\ \eqref{eq:lower} (the lower tail bound in  Theorem \ref{thm:concent}), we have 
$$\PR\big(W_{i,j}^k(t)\big) =\PR\Big( A_{i,j}^k(t)\leq \E\big[A_{i,j}^k(t)\big] -x\Big)
\leq \exp\Big\{-\frac{x^2}{2\E[A_{i,j}^k(t)]}\Big\}.$$
We note that $\E[A_{i,j}^k(t)]=\rho t/n\leq b/n=c_b f_n^2 (\log f_n)/n$. Using also Eq.\ \eqref{eq:x}, we obtain 
$$\frac{x^2}{2\E[A_{i,j}^k(t)]}\geq \frac{c_d^2 f_n^2 \log^2 f_n}{16n}\cdot
\frac{1}{2c_b f_n^2 (\log f_n)/n}=\frac{c_d^2}{32 c_b}\log f_n \geq 20 \log f_n,$$
where the last inequality follows from our assumption that $c_d^2\geq 640 c_b$; cf.~Eq.~\eqref{eq:c}. Consequently,
$$\PR\big(W_{i,j}^k(t)\big)\leq \exp\{-20\log f_n\}=\frac{1}{f_n^{20}}\leq\frac{1}{2f_n^{19}}.$$

The event $W^k$ is the union of {$n^2(b-d)$} events $W^k_{i,j}(t)$. We note that
\begin{equation}\label{eq:nbb}
{n^2(b-d)}\leq n^2b \leq f_n^2 c_b f_n^2 \log f_n
\leq f_n^6,
\end{equation}
as long as $n$ is large enough so that $c_b\leq f_n$.
Therefore, using the union bound
$$\PR(W^k)\leq {n^2(b-d)} \frac{1}{2f_n^{19}} \leq \frac{f_n^6}{2f_n^{19}}=
\frac{1}{2f_n^{13}}.$$
\end{proof}

\subsection{The probability of no new backlog}
In this subsection we show that $U_k$, the additional backlog generated during the $k$th service period, is zero with high probability. Our analysis builds on an upper bound on the probability that the number of packets in the $k$th batch that are associated with a particular port is appreciably larger than its expected value. Towards this purpose, we define the row and column sums for the arrivals in the $k$th batch:
$$R_i^k=\sum_j A^k_{i,j}(b),\qquad C_j^k =\sum_i A^k_{i,j}(b).$$
We also define  the events
$$F_i^k=\{R_i^k>s\},\qquad G^k_j=\{C_j^k>s\},$$
and  
$$H^k=\big( F_1^k\cup\cdots \cup F_n^k\big) \cup \big( G_1^k\cup\cdots \cup G_n^k\big).$$

In what follows, we first show that the event $H^k$ has low probability. We then show that if neither of the events $W^k$ or $H^k$ occurs (which has high probability), then $U_k$ is equal to zero.

\begin{lemma}\label{l:h}
For $n$ sufficiently large,  we have
$$\PR(H^k)\leq \frac{1}{2f_n^{13}},\qquad\mbox{for all }k.$$
\end{lemma}
\begin{proof}
Let us focus on the event $F_1^k=\{R_1^k>s\}$; the argument for other events $F_i^k$ or $G_j^k$ is identical. 
Note that $\E[R_1^k]=\rho b$.
We have, using Eq.\ \eqref{eq:upper} (the upper tail bound in  Theorem \ref{thm:concent}) in the last step,
\begin{eqnarray*}
\PR(R_1^k>s)&=&\PR\big(R_1^k >\rho b +\sqrt{c_s b \log f_n}\big)\\
&=&\PR\big(R_1^k >\E[R_1^k] +\sqrt{c_s b \log f_n}\big)\\
&\leq&\exp\Big\{-\frac{c_s b\log f_n}{2(\rho b+x/3)}\Big\},
\end{eqnarray*}
where $x=\sqrt{c_s b \log f_n}$. Notice that
$$\rho b+\frac{x}{3}\leq \rho b +x = \rho b+\sqrt{c_s b \log f_n}=s\leq b.$$
Therefore, {when $n \geq 4$,}
$$\PR(R_1^k>s)\leq \exp\Big\{-\frac{c_s b\log f_n}{2b}\Big\}=\frac{1}{f_n^{c_s/2}}
\leq \frac{1}{4f_n^{14}},$$
where the last inequality follows from our assumption that $c_s\geq 30$; cf.~Eq.~\eqref{eq:c}.
The event $H^k$ is the union of $2n$ events, each with probability bounded above by $1/(4f_n^{14})$. Using the union bound and the assumption $n\leq f_n$, we obtain $\PR(H^k)\leq 1/(2f_n^{13})$.
\end{proof}

\begin{lemma}\label{wh}
\mbox{ }
\begin{enumerate}
\item[(a)] Consider a sample path under which neither $W^k$ nor $H^k$ occurs. Then, $U_k=0$.
\item[(b)] We have $\PR(U_k>0)\leq 1/f_n^{13}$.
\item[(c)] {For every sample path,} we have $U_k\leq n^2b$.
\end{enumerate}
\end{lemma}

\begin{proof}
\begin{itemize}
\item[(a)] We assume that neither $W^k$ nor $H^k$ occurs.
Using Eq.~\eqref{eq:qk}, the queue sizes (where we only count packets from the $k$th batch)  at the beginning of the normal clearing period are equal to 
\begin{equation}\label{eq:hat}
Q^k_{i,j}(b+1)=A_{i,j}^k(b)-S^k_{i,j}(b).\end{equation}
Let
$$\hat R_i^k =\sum_j Q^k_{i,j}(b+1),\qquad \hat C_j^k =\sum_i Q^k_{i,j}(b+1).$$
{Now consider a fixed $i$.}
Note that the schedules $\bsigma^{(m)}$ used during the round-robin phase have the property $\sum_j \sigma_{i,j}^{(m)}=1$; that is, each input port is offered exactly one unit of service at each time slot. Furthermore, since event $W^k$ does not occur, Lemma~\ref{l:aa} implies that all queues are positive 
{at the beginning of each slot of} 
the round-robin phase; {that is, $Q^k_{ij}(t+1)>0$, for $t=d,\ldots,b-1$. Therefore, the offered service is never} wasted during the $b-d$ slots of the round-robin phase. It follows that
the total actual service at input port $i$ during the round-robin phase is exactly $b-d$:
$$\sum_j S_{i,j}^k(b)=b-d.$$
Furthermore, since event $H^k$ does not occur, we have $R_i^k\leq s$. 
{Recalling the definition $R_i^k = \sum_j A_{i,j}^k(b)$, and by}
summing both sides of Eq.~\eqref{eq:hat} over all $j$, we obtain
$$\hat R_i^k = R_i^k-\sum_j S_{i,j}^k(b)\leq s-(b-d)=\ell,$$
where $\ell=d+s-b$ is the length of the normal clearing phase.
By a similar argument, we obtain that $\hat C_j^k\leq \ell$, for all $j$.
It then follows from Theorem \ref{thm:bvn} that all the packets (from the $k$th arrival batch) will be cleared during the normal clearing phase, and $U_k=0$.

\item[(b)] If $U_k>0$, then, by part (a), it must be that either event $W^k$ or $H^k$ occurs. The result follows because the probability of each one of these two events is upper bounded by $1/(2f_n^{13})$ (Lemmas \ref{l:w} and \ref{l:h}).

\item[(c)] The number of packets from the $k$th batch that can get backlogged can be no more than the total number of arrivals in the $k$th batch. Since each queue ($n^2$ of them) receives at most one packet at each time slot ($b$ slots), the total number cannot exceed $n^2 b$.
\end{itemize}
\end{proof}

\subsection{Backlog analysis} 
We are now in a position to show that the expected backlog is very small.

\begin{lemma}\label{l:back}
Assuming that $n$ is sufficiently large, 
we have that $\E[B_k]\leq 1$, for all $k$.
\end{lemma}
\begin{proof}
Using Eq.~\eqref{eq:bevol}, 
the backlog satisfies
$$B_{k+1}\leq\max\{0,B_k+U_k-r\}\leq \max\{0,B_k+U_k-1\}.$$
Let us define a sequence $\hat B_k$ with the recursion $\hat B_0=0$ and
$$\hat B_{k+1}= \max\{0,\hat B_k+U_k-1\}.$$
We then have $B_k\leq \hat B_k$, so it suffices to derive an upper bound on $\E[\hat B_k]$. 

We use the discrete-time Kingman bound (Theorem \ref{thm:kingman}), where we identify $Z(\tau)$ with $\hat B_k$, $X(\tau)$ with $U_k$, and $Y(\tau)$ with 1.  Using the notation in Theorem \ref{thm:kingman}, 
we have $\mu = 1$, and $m_{2y} = 1$. Furthermore, as in Eq.~\eqref{eq:nbb}, we have
$n^2b\leq f_n^6$
for sufficiently large $n$. Using Lemma \ref{wh},
$$\lambda = \E[U_k] \leq f_n^6 \cdot \PR(U_k > 0) \leq f_n^{6}\cdot\frac{1}{f_n^{13}} = \frac{1}{f_n^{7}},$$ 
{and} $$m_{2x} = \E[U_k^2] \leq f_n^{12} \cdot \PR(U_k > 0) = f_n^{12}\cdot\frac{1}{f_n^{13}} = \frac{1}{f_n}.$$
Then, using the bound in \eqref{eq:kingman-discrete}, we have 
\[
\E[B_k] \leq \E[\hat B_k] \leq \frac{m_{2x} + m_{2y}}{2(\mu - \lambda)} \leq \frac{f_n^{-1}+1} {2(1-f_n^{-7})}.
\]
As $n$ increases, the right-hand side converges to $1/2$ and is therefore bounded above by $1$ when $n$ is
 sufficiently large.
\end{proof}

\subsection{Queue size analysis} 

In this subsection we show that at any time, the sum of the queue sizes is of order $O(nd)$. We fix some time $\tau$ and consider two cases, depending on whether this time belongs to a round-robin phase or not.

\paragraph{Queue sizes during the round-robin phase} 
Suppose that $\tau$ satisfies $kb+d+1\leq \tau\leq (k+1)b$, so that $\tau$ belongs to the round-robin phase of the $k$th service period, and let us look at the queue size $Q_{i,j}(\tau+1)$. 
This queue size may include some packets that arrived during earlier arrival periods and that were backlogged; {their} total expected number (summed over all $i$ and $j$) is $\E[B_k]\leq 1$. 

Let us now turn our attention to packets that belong to the $k$th batch.  Recall that the number of such packets in queue $(i,j)$ at the beginning of the $(t+1)$st slot (equivalently, the end of the $t$th slot) of the $k$th arrival period  is denoted by $Q^k_{i,j}(t+1)$. For $t=d+1,\ldots, b$,  we have, as in Eq.~\eqref{eq:qk},
$$Q^k_{i,j}(t+1)=A^k_{i,j}(t)-S_{i,j}^k(t),$$
and
$$\sum_{i,j}\E[Q^k_{i,j}(t+1)]=n\rho t-\E\Big[\sum_{i,j} S_{i,j}^k(t)\Big].$$
By the same argument as in the proof of Lemma \ref{wh}(a), if event {$W^k$} does not occur,  the service during the round-robin phase is never wasted: a total of $n$ packets are served at each time, and {for $t=d+1,\ldots,b$,} a total of $n(t-d)$ packets are served by the $t$th slot of the $k$th arrival period. Using also the inequality (cf.\ Lemma\ \ref{l:w})
$$1-\PR({W^k})\geq 1-\frac{1}{{2}f_n^{13}}\geq 1-\frac{1}{f_n}=\rho,$$ we obtain
$$\E\Big[\sum_{i,j} S_{i,j}^k(t)\Big]\geq n(t-d)\big(1-\PR({W^k})\big)\geq n \rho(t-d).$$
Therefore,
\begin{equation}\label{eq:qqq}
\sum_{i,j}\E[Q^k_{i,j}(t+1)]\leq n\rho t - n \rho(t-d)=n\rho d\leq nd,\qquad t=d+1,\ldots,b.
\end{equation}
which is an upper bound of the desired form.

\paragraph{Queue sizes outside the round-robin phase} 

Suppose now that $\tau$ satisfies $(k+1)b+1\leq \tau\leq (k+1)b+d$, so that $\tau$ belongs to one of the last two phases of the $k$th service period, and let us look again at the queue size $Q_{i,j}(\tau+1)$. 
As before, we may have some backlogged packets. These are either packets backlogged during the current period (the $k$th one) or in previous periods. Their total expected number (summed over all $i$ and $j$) at any time in this range
 is upper bounded by $\E[B_k+U_k]\leq 2$. 

Let us now turn our attention to packets that belong to the $k$th batch.
Since there are no further arrivals from the $k$th batch from slot $(k+1)b+1$ onwards, the number of such packets is largest at the beginning of slot $(k+1)b+1$.
Their expected value {at that time satisfies}
$$\sum_{i,j}\E\big[Q^k_{i,j}(b+1)\big]\leq nd,$$
where in the  inequality we used Eq.~\eqref{eq:qqq} with $t=b$.

Finally, we need to account for arrivals that belong to the $(k+1)$st arrival batch.
The total number of such accumulated arrivals is largest when we consider the largest value of $\tau$, namely, $\tau=(k+1)b+d$. By that time, we have had a total of $d$ slots of the $(k+1)$st arrival {period}, and a total expected number of arrivals equal to $\rho nd$, which is bounded above by $nd$.

Putting together all of the bounds that we have developed, we see that at any time, the expected total number of packets is bounded above by $2nd +2\leq 3nd$. This being true for all sufficiently large $n$,  establishes Theorem \ref{thm:main}.

\section{Discussion}\label{sec:discussion}
We presented a novel scheduling policy 
for an $n\times n$ input-queued switch. 
In the regime where the system load \jnt{satisfies} $\rho = 1 - 1/n$, 
and the arrival rates \jnt{at the different queues are all equal,}
our policy achieves an upper bound of order $O(n^{2.5} \log n)$
on the \jnt{expected} total queue size, 
a substantial improvement 
upon \jnt{earlier} upper bounds, all of which \jnt{were} of order $O(n^3)$, 
ignoring poly-logarithmic dependence on $n$. \jnt{Our policy is of the batching type. However, instead of waiting until an entire batch has arrived, our policy only waits for 
enough arrivals to take place for the system to exhibit a desired level 
of regularity, and then starts serving the batch.
This idea may be of independent interest.} 

\jnt{Our} policy uses detailed knowledge of the arrival statistics, 
and is heavily dependent on the fact that \jnt{all} arrival rates are \jnt{the same.} 
While we believe that similar policies can be devised 
for arbitrary arrival rates \jnt{(within the regime considered in this paper),}
the policy description and analysis are likely to be more involved. 

\jnt{Finally, for} the regime where $\rho \approx 1 - 1/n$, 
\jnt{there is a} $\Omega(n^2)$ lower bound 
on the \jnt{expected} total queue size \jnt{under any policy}
(see \cite{STZopen}), 
whereas our upper bound is of order $O(n^{2.5} \log n)$. 
It is \jnt{an interesting open question} whether this gap between the upper and lower bound 
can be closed. \jnt{Our policy uses a prespecified sequence of schedules (round-robin) until the entire batch has arrived and then uses an ``adaptive'' sequence of schedules to clear remaining packets after the end of the batch. Within the class of policies of this type, with perhaps different choices of the parameters involved, it appears to be impossible to obtain an upper bound of $O(n^{\alpha})$ for $\alpha<2.5$. Thus, in order to come closer to the $\Omega(n^2)$ lower bound, we will have to use an adaptive 
sequence of schedules early on, before the entire batch has arrived. In fact, if one were to achieve an upper bound close to $O(n^2)$, we would have an approximately constant expected number of packets in each queue. This means that with positive probability, many of the queues will be empty. Therefore, an elaborate policy would be needed to avoid offering service to empty queues and thus avoid queue buildup.  But the analysis of such elaborate policies appears to be a difficult challenge.}

\bibliographystyle{spmpsci}      
\bibliography{bibliography}   


\end{document}